%% file: quarter_bound.tex
\documentclass[11pt]{article}
\usepackage[T1]{fontenc}
\usepackage[utf8]{inputenc}
\usepackage{lmodern}
\usepackage[margin=1in]{geometry}
\usepackage{amsmath,amssymb,amsthm,mathtools}
\usepackage{microtype,booktabs,longtable,array,enumitem,fancyvrb,xurl}
\usepackage[hidelinks]{hyperref}
\hypersetup{pdftitle={A New Upper Bound on the Binary Deletion Channel Capacity},
  pdfauthor={Özgür Soysal},
  pdfsubject={Binary deletion channel capacity converse with a verified Lean certificate},
  pdfkeywords={binary deletion channel, capacity, Lean, finite certificate}}
\numberwithin{equation}{section}
\newtheorem{theorem}{Theorem}[section]
\newtheorem{lemma}[theorem]{Lemma}
\newtheorem{proposition}[theorem]{Proposition}
\newtheorem{corollary}[theorem]{Corollary}
\theoremstyle{definition}
\newtheorem{definition}[theorem]{Definition}
\newtheorem{remark}[theorem]{Remark}
\DeclareMathOperator{\Ent}{H}
\DeclareMathOperator{\Div}{D}
\DeclareMathOperator{\supp}{supp}
\newcommand{\E}{\mathbb E}
\newcommand{\Prob}{\mathbb P}
\newcommand{\R}{\mathbb R}
\newcommand{\N}{\mathbb N}
\newcommand{\U}{\mathcal U}
\newcommand{\Y}{\mathcal Y}
\newcommand{\ip}[2]{\langle #1,#2\rangle}

\newcommand{\one}{\mathbf 1}
\newcommand{\bits}{\{0,1\}}
\newcommand{\lean}[1]{\par\smallskip\noindent\textit{Formal correspondence:}
  {\small\path{#1}}.\par\smallskip}
\setlist[enumerate]{itemsep=3pt,topsep=5pt}

\title{A New Upper Bound on the Binary Deletion Channel Capacity}
\author{\"Ozg\"ur Soysal}
\date{September 11, 2026}

\begin{document}
\maketitle
\begin{center}
\begin{minipage}{0.9\textwidth}
\centering
\textbf{AI: disclosure}\par\smallskip
\small This work is an AI-assisted mathematics project. Claude Fable 5.1,
GPT 5.6 Sol, and GPT 6 Astra were used in its development.
The author is responsible for the manuscript and its claims.
\end{minipage}
\end{center}
\begin{abstract}
We prove that the capacity of the binary deletion channel satisfies
$C(d)\le (1-d)/4$ for every $13/20\le d<1$. The proof describes the
output from right to left, using a six-bit context to assign a description
length. We bound the increase in expected description length minus output
entropy when one input bit is added. A relative-entropy identity reduces
this bound to finitely many linear inequalities. A potential on input
windows of length 26 makes the inequalities telescope, giving a bound for
every input word. Deletion composition extends the result from $d=13/20$
to all larger deletion probabilities. We also obtain finite-block bounds
on mutual information and decoding error, with explicit
$O(\log n/n)$ corrections. The finite certificate is checked using exact
integer arithmetic, and the proof is formalized end to end in Lean.
\end{abstract}


\tableofcontents
\newpage

\section{Introduction and statement of results}\label{sec:intro}

The binary deletion channel independently removes each input bit with
probability $d$ and delivers the surviving subsequence in its original order.
The receiver observes neither the deleted bits nor their positions. The
resulting uncertainty about alignment makes even this elementary channel
difficult to analyze. Cheraghchi and Ribeiro~\cite{CR} survey the capacity
theory of synchronization channels and the methods developed for them.

The present paper gives a converse in the high-deletion regime by following
the probability law of a short prefix of the future output. This law retains
enough alignment information to control the increase of output entropy when
one input bit is prepended. A backward description cost uses the same law.
Their difference is a local reward. We dominate this reward by a finite
family of affine functionals and certify a potential inequality over every
input-window state. The finite certificate controls all remote suffix laws,
not just the reference distributions used to construct it.

The result is stated operationally, for uniform-message average-error codes
with arbitrary deterministic encoders and decoders. Section~\ref{sec:model}
fixes these definitions. Entropy and mutual information in displayed
information expressions are measured in nats; rates and capacity are in bits
per input bit.

\begin{theorem}[Quarter bound and finite-block converses]\label{thm:main}
For every $d\in[13/20,1)$ and every positive integer $n$:
\begin{enumerate}[label=(\roman*)]
\item Every law of $X\in\bits^n$, with channel output $Y$, satisfies
\begin{equation}\label{eq:main-info}
 \frac{I(X;Y)}{n\log 2}
 \le \frac{1-d}{4}+\frac{\log_2(n+1)+1}{n}.
\end{equation}
\item Every length-$n$ block code of rate $R_n$ and average error $P_e$
satisfies
\begin{equation}\label{eq:main-code}
 (1-P_e)R_n
 \le \frac{1-d}{4}+\frac{\log_2(n+1)+2}{n}.
\end{equation}
\item The operational capacity satisfies
\begin{equation}\label{eq:main-capacity}
 C(d)\le\frac{1-d}{4}.
\end{equation}
\item If $n_j\to\infty$ and length-$n_j$ codes have $P_{e,j}\to0$, then
for every $\varepsilon>0$, all sufficiently large $j$ satisfy
\begin{equation}\label{eq:main-eventual}
 R_{n_j}<\frac{1-d}{4}+\varepsilon.
\end{equation}
In particular, any limit of those rates is at most $(1-d)/4$.
\end{enumerate}
\end{theorem}

The proof has two parts. First, Theorem~\ref{thm:generic} shows how a
finite collection of inequalities bounds the capacity. Then
Proposition~\ref{prop:concrete} supplies an explicit certificate satisfying
those inequalities. The argument leading from the certificate to the
capacity bound is given in full below.

\begin{corollary}\label{cor:limsup}
The high-deletion normalized capacity obeys
\[
 \limsup_{d\uparrow1}\frac{C(d)}{1-d}\le\frac14.
\]
\end{corollary}
\begin{proof}
For $13/20\le d<1$, divide \eqref{eq:main-capacity} by the positive
number $1-d$ and take the upper limit.
\end{proof}
This corollary does not require an existence theorem for the normalized
limit. Dalai~\cite{Dalai} proved that the limit exists and equals the
infimum of the normalized capacity; invoking that external theorem gives
the same upper bound on the limit itself.

\subsection{Relation to previous bounds}
Fertonani and Duman~\cite{FD} introduced the genie-aided auxiliary-channel
framework underlying the later upper bounds of Rubinstein and
Con~\cite{RC} and Pinto and Ribeiro~\cite{PR}. By providing side
information about the deletion process, their framework bounds the
deletion-channel capacity through the capacities of finite auxiliary
channels, which can be computed using the Blahut--Arimoto
algorithm~\cite{Blahut,Arimoto}. Building on this framework, Rubinstein
and Con~\cite{RC} obtained the upper bound $0.3745(1-d)$ for $d\ge0.68$
by reducing the memory requirements of these computations. Pinto and
Ribeiro~\cite{PR} further developed the same approach with a parallelized
implementation, obtaining $0.3578(1-d)$ for $d\ge0.64$.
The quarter bound improves the Rubinstein--Con coefficient for
$d\ge0.68$ and the Pinto--Ribeiro coefficient for $d\ge0.65$. The present construction uses an input-window state
space and a fixed predictive output context, rather than the entire
transition matrix of a finite-block subsequence channel.

Auxiliary output distributions are standard tools in divergence-based
capacity converses; see, for example, the synchronization-channel
applications surveyed in~\cite[Sec.~4]{CR}. Here the specific backward
output law and the global finite certificate give a bound uniform over
all input words and all remote suffix laws.

\subsection{Outline of the proof}\label{subsec:scope}
The main quantity is expected description length minus output entropy.
Section~\ref{sec:testlaw} explains why an upper bound on this quantity gives
a capacity converse. Sections~\ref{sec:reward} and~\ref{sec:tangent} bound
its increase when a bit is prepended. Section~\ref{sec:window} reduces the
problem to a finite inequality and shows how a potential makes the local
bounds telescope. Section~\ref{sec:converse} then derives the information,
decoding, and capacity bounds.

The remaining step is an explicit computer-assisted calculation.
Sections~\ref{sec:integer} and~\ref{sec:coverage} specify the integer
inequalities and the certificate that satisfies them. Its potential has
$2^{26}$ entries and is supplied electronically. The complete proof,
including this finite calculation, is formalized in Lean; the appendices
give the implementation details and verification record.

\section{Channel model, finite laws, and operational capacity}\label{sec:model}

\subsection{Words and the channel kernel}
Let $\bits^0=\{\varnothing\}$ and
$\Y_n=\bigsqcup_{k=0}^n\bits^k$; the disjoint union records the word length.
For $x=x_1\cdots x_n$ and $e\in\bits^n$, let $x_e$ be the subsequence
of coordinates with $e_i=1$, and let $|e|=\sum_i e_i$. Define
\begin{equation}\label{eq:channel}
 W_d(y\mid x)=\sum_{e\in\bits^n:\,x_e=y}
                 (1-d)^{|e|}d^{n-|e|},\qquad y\in\Y_n.
\end{equation}
As usual, exponent zero contributes the factor one, including at an
endpoint of $[0,1]$. All summands are nonnegative and
$\sum_yW_d(y\mid x)=(d+(1-d))^n=1$. Distinct masks yielding the same
word are summed, not identified with a single mask.

For a finite set $A$, write $\mathcal P(A)$ for its probability simplex.
For $P\in\mathcal P(A)$ and a map $f:A\to B$, the pushforward is
$(f_{\#} P)(b)=\sum_{a:f(a)=b}P(a)$. We use
\[
 \Ent(P)=-\sum_{a:P(a)>0}P(a)\log P(a),\qquad
 \Div(P\Vert Q)=\sum_{a:P(a)>0}P(a)\log\frac{P(a)}{Q(a)}
\]
when $\supp P\subseteq\supp Q$. No expression $\log0$ is needed in this
convention. All uses of divergence below either have a strictly positive
second law or explicitly verify the support inclusion.

If $X$ has law $P$ and channel $W$, its output law is $P_Y=\sum_xP(x)W_x$,
and
\[
 I(X;Y)=\sum_xP(x)\Div(W_x\Vert P_Y)
       =\Ent(P_Y)-\sum_xP(x)\Ent(W_x).
\]
Terms with $P(x)=0$ can be omitted, so the divergence support condition
holds wherever it is used.

\subsection{Codes and achievable rates}
A block code of input length $n$ is a positive integer $M$, an encoder
$f:\{1,\ldots,M\}\to\bits^n$, and a decoder
$g:\Y_n\to\{1,\ldots,M\}$. The encoder need not be injective.
For a uniform message $U$, set $X=f(U)$ and
\[
 P_e=\Prob\{g(Y)\ne U\},\qquad R_n=\frac{\log M}{n\log2}\quad(n>0).
\]
A nonnegative rate $R$ is achievable if there is such a code at every
input length, its error tends to zero, and its rate tends to $R$.
We define $C(d)$ as the supremum of these rates. The one-message code
shows that the achievable set contains zero. The upper bound proved
below shows that this set is bounded above.

Our finite-block statements concern positive lengths. The conclusion
\eqref{eq:main-eventual} also applies to codes defined along any sequence
of lengths tending to infinity, so the converse does not depend on
requiring a code at every length.

\section{Finite entropy inequalities used in the proof}\label{sec:entropy}

The chain rule, log-sum inequality, convexity of relative entropy, and
data-processing inequality are classical; see Cover and
Thomas~\cite[Ch.~2]{CT}. We include their finite-law proofs to make the
support conditions explicit.

For finite random variables $A,B$, define
\[
 \Ent(A\mid B)=\sum_{b:\,\Prob(B=b)>0}\Prob(B=b)
                   \Ent(\mathcal L(A\mid B=b)).
\]
Conditional laws on zero-probability events do not enter this sum.
On every positive joint-mass coordinate,
$\Prob(A=a,B=b)=\Prob(B=b)\Prob(A=a\mid B=b)$.
Taking logarithms, multiplying by joint mass, and summing gives the
chain rule $\Ent(A,B)=\Ent(B)+\Ent(A\mid B)$.
Subtracting the two chain-rule forms gives
$I(A;B)=\Ent(A)-\Ent(A\mid B)=\Ent(B)-\Ent(B\mid A)$.
The same calculation conditional on a third variable gives the
conditional chain rule used in the decoding argument.

\begin{lemma}[Log-sum inequality with support conditions]\label{lem:logsum}
Suppose $a_i,b_i\ge0$ on a finite index set, and $b_i=0$ implies $a_i=0$.
Put $A=\sum_i a_i$ and $B=\sum_i b_i$. If $A>0$, then $B>0$ and
\begin{equation}\label{eq:logsum}
 A\log(A/B)\le\sum_{i:a_i>0}a_i\log(a_i/b_i).
\end{equation}
If $A=0$, both sides are interpreted as zero and the inequality remains true.
\end{lemma}
\begin{proof}
If $A=0$, every $a_i$ vanishes. Otherwise $B>0$ by the support condition.
Set $p_i=a_i/A$ and $q_i=b_i/B$. For $p_i>0$, the elementary inequality
$\log t\le t-1$, applied to $t=q_i/p_i$, gives
$p_i\log(p_i/q_i)\ge p_i-q_i$. Summing over $p_i>0$ gives a lower bound
$1-\sum_{p_i>0}q_i\ge0$. Multiplication by $A$ and expansion of the
logarithm give \eqref{eq:logsum}.
\end{proof}

\begin{lemma}[Divergence contraction and mixing]\label{lem:kl}
Under the required support inclusions:
\begin{enumerate}[label=(\roman*)]
\item $\Div(P\Vert Q)\ge0$;
\item $\Div(f_{\#} P\Vert f_{\#} Q)\le\Div(P\Vert Q)$ for every map $f$;
\item for probability weights $\lambda_j$,
\[
 \Div\left(\sum_j\lambda_jP_j\middle\Vert\sum_j\lambda_jQ_j\right)
 \le\sum_j\lambda_j\Div(P_j\Vert Q_j).
\]
\end{enumerate}
\end{lemma}
\begin{proof}
Part (i) is Lemma~\ref{lem:logsum} with totals one. For (ii), apply that
lemma on each fiber $f^{-1}(b)$ and sum over $b$. For (iii), at each
coordinate apply it to $a_j=\lambda_jP_j(a)$ and
$b_j=\lambda_jQ_j(a)$, then sum over $a$. A zero mixture weight contributes
zero and causes no exceptional case.
\end{proof}

In particular, a law on an alphabet of size $k\ge1$ has entropy at
most $\log k$: its divergence from the uniform law is
$\log k-\Ent(P)\ge0$. Applying this to each conditional law also
bounds conditional entropy by $\log k$. If a random variable is
determined by the variables being conditioned on, that conditional
entropy is zero because every contributing law is a point mass.

\begin{lemma}[Information and deterministic processing]\label{lem:mi}
For finite random variables $B,Z$ and a function $v$,
$I(B;v(Z))\le I(B;Z)$. Also, if $Q$ is a positive law on the output
alphabet of a finite channel, then
\begin{equation}\label{eq:dualidentity}
 \sum_xP(x)\Div(W_x\Vert Q)=I(X;Y)+\Div(P_Y\Vert Q).
\end{equation}
\end{lemma}
\begin{proof}
For the first assertion, apply Lemma~\ref{lem:kl}(ii) to the joint law of
$(B,Z)$ and the product of its marginals, under $(b,z)\mapsto(b,v(z))$.
The support condition holds since positive joint mass implies positive
marginal masses. For the second assertion, expand the left side and add
and subtract $\sum_yP_Y(y)\log P_Y(y)$. Coordinates with $P_Y(y)=0$
have zero mass in every positively weighted $W_x$, so the expansion is
valid on the indicated supports.
\end{proof}

\section{Backward description lengths and a positive test law}\label{sec:testlaw}

We first construct a reference distribution on output words. Its negative
logarithm is a description length, up to a normalizing constant. Bounding
expected description length minus entropy will therefore bound divergence
from this reference distribution.

Let $\U=\{0,1,\ldots,63\}$. For $c\in\bits$, define
\[
 \sigma_c(i)=(2i+c)\bmod64.
\]
Define the context recursively by
$u(\varnothing)=0$ and $u(cy)=\sigma_c(u(y))$. Equivalently, for
$y=y_1\cdots y_k$,
\begin{equation}\label{eq:context}
 u(y)=\sum_{j=1}^{\min(k,6)}2^{j-1}y_j.
\end{equation}
Missing coordinates are zero. This statistic may identify a short word
with a longer word ending in zeros; injectivity of $u$ is not assumed.

The Kraft inequality is the standard normalization condition for
prefix-code lengths~\cite[Ch.~5]{CT}. We use its real-length form:
fix $w_{i,c}$ satisfying the row-wise inequalities
\begin{equation}\label{eq:kraft}
 e^{-w_{i,0}\log2}+e^{-w_{i,1}\log2}\le1
 \qquad(i\in\U).
\end{equation}
Define $L(\varnothing)=0$ and
\begin{equation}\label{eq:length}
 L(cy)=w_{u(y),c}+L(y).
\end{equation}
These real description lengths need not be lengths of literal binary
codewords. Only the positive weights $2^{-L(y)}$ and their partition
bound enter the proof.

\begin{lemma}[Partition bound]\label{lem:partition}
For $Z_n=\sum_{y\in\Y_n}2^{-L(y)}$, one has $1\le Z_n\le n+1$.
Thus $Q_n(y)=2^{-L(y)}/Z_n$ is a strictly positive probability law on
$\Y_n$ and, for every $P\in\mathcal P(\Y_n)$,
\begin{equation}\label{eq:testdiv}
 \Div(P\Vert Q_n)
 \le (\log2)\E_P L-\Ent(P)+\log(n+1).
\end{equation}
\end{lemma}
\begin{proof}
The empty word contributes one. Every nonempty word of length at most
$n+1$ is uniquely $cy$ with $y\in\Y_n$. Therefore
\[
 Z_{n+1}=1+\sum_{y\in\Y_n}2^{-L(y)}
       \left(2^{-w_{u(y),0}}+2^{-w_{u(y),1}}\right)\le1+Z_n.
\]
Induction from $Z_0=1$ gives the upper bound. The divergence identity
\[
 \Div(P\Vert Q_n)=(\log2)\E_P L-\Ent(P)+\log Z_n
\]
then proves \eqref{eq:testdiv}. In particular, the test law includes the
empty output and every possible output length.
\end{proof}

\section{Predictive laws and the entropy increment}\label{sec:reward}

Fix temporarily $0<d<1$. For a law $\mu$ on $\U$, put
\begin{equation}\label{eq:filter}
 F_c\mu=d\mu+(1-d)(\sigma_c)_{\#}\mu.
\end{equation}
For a word $x=x_1\cdots x_k$, set
$F_x=F_{x_1}\circ\cdots\circ F_{x_k}$, with the rightmost operator
acting first; $F_\varnothing$ is the identity. Let $p_x$ be the law of
$u(Y)$ under $W_d(\cdot\mid x)$.

\begin{lemma}[Exact recursion]\label{lem:recursion}
$p_\varnothing=\delta_0$, $p_{cx}=F_cp_x$, and $p_x=F_x\delta_0$.
\end{lemma}
\begin{proof}
Let $Y$ be the deletion output from $x$, and let $B$ independently equal
one with probability $1-d$. The deletion output from $cx$ is $Y$ if
$B=0$ and $cY$ if $B=1$. Equation~\eqref{eq:context} gives the
recursion. Induction proves the final identity.
\end{proof}

Define the weighted Jensen--Shannon divergence~\cite{Lin}
\[
 J_d(\mu,\nu)=\Ent(d\mu+(1-d)\nu)-d\Ent(\mu)-(1-d)\Ent(\nu).
\]
It is the mutual information between a binary mixture label and the
resulting observation. Define the score and local reward, in nats, by
\begin{align}
 S_d(x)&=(\log2)\E[L(Y)\mid x]-\Ent(W_d(\cdot\mid x)),\label{eq:score}\\
 r_c(\mu)&=(1-d)(\log2)\sum_i\mu_iw_{i,c}
             -J_d(\mu,(\sigma_c)_{\#}\mu).\label{eq:reward}
\end{align}

\begin{lemma}[One-symbol score bound]\label{lem:increment}
$S_d(\varnothing)=0$ and $S_d(cx)-S_d(x)\le r_c(p_x)$.
\end{lemma}
\begin{proof}
Use the variables $Y,B$ in Lemma~\ref{lem:recursion}, and call the new
output $Z$. Conditional on $B=0$, the map $Y\mapsto Z$ is the identity;
conditional on $B=1$, it prepends $c$. Both are injective on the old
output alphabet, so $\Ent(Z\mid B)=\Ent(Y)$. Consequently
\[
 \Ent(Z)-\Ent(Y)=I(B;Z)\ge I(B;u(Z))
     =J_d(p_x,(\sigma_c)_{\#} p_x).
\]
Independence and \eqref{eq:length} give
$\E L(Z)-\E L(Y)=(1-d)\sum_i(p_x)_iw_{i,c}$. Subtract the entropy
inequality from this cost identity. For the empty input both the output
entropy and its description length are zero.
\end{proof}

\begin{remark}
The laws of $Y$ and $cY$ can overlap when embedded in the larger output
alphabet. The proof does not add a binary entropy term as though their
supports were disjoint. It uses the actual mutual information of the
retention label. This distinction is necessary for the inequality to hold.
\end{remark}

\section{A global affine majorant}\label{sec:tangent}

The reward depends nonlinearly on the predictive law. We now bound it by a
linear function of that law. This will allow us to check the bound on a
finite collection of distributions.

Let $q\in\mathcal P(\U)$ be strictly positive. Define a vector
$g_c(q)\in\R^{64}$ by
\begin{align}
 g_c(q)_i={}&(1-d)(\log2)w_{i,c}
       +d\log(F_cq)_i+(1-d)\log(F_cq)_{\sigma_c(i)}\notag\\
       &-d\log q_i
       -(1-d)\log\bigl((\sigma_c)_{\#} q\bigr)_{\sigma_c(i)}.
                                                        \label{eq:g}
\end{align}
Every logarithm displayed here is of a positive number: $F_cq\ge dq>0$,
and $((\sigma_c)_{\#} q)_{\sigma_c(i)}\ge q_i>0$.

\begin{lemma}[Exact majorant slack]\label{lem:slack}
For every $\mu\in\mathcal P(\U)$,
\begin{align}
 \ip{\mu}{g_c(q)}-r_c(\mu)
   ={}&d\Div(\mu\Vert q)
    +(1-d)\Div((\sigma_c)_{\#}\mu\Vert(\sigma_c)_{\#} q)\notag\\
    &-\Div(F_c\mu\Vert F_cq)\ge0.\label{eq:slack}
\end{align}
In particular, $r_c(\mu)\le\ip{\mu}{g_c(q)}$, with equality at $\mu=q$.
\end{lemma}
\begin{proof}
The support condition for the second divergence holds because a zero
coordinate of $(\sigma_c)_{\#} q$, with $q$ positive, has an empty
preimage under $\sigma_c$; the corresponding pushforward of $\mu$ is
also zero. The other denominator laws are positive.

Write $\nu=(\sigma_c)_{\#}\mu$ and $v=(\sigma_c)_{\#} q$.
By the pushforward summation identity, expansion of \eqref{eq:g} yields
\begin{align*}
 \ip{\mu}{g_c(q)}={}&(1-d)(\log2)\ip{\mu}{w_{\cdot,c}}
 +\sum_i(F_c\mu)_i\log(F_cq)_i\\
 &-d\sum_i\mu_i\log q_i
 -(1-d)\sum_{i:\,\nu_i>0}\nu_i\log v_i.
\end{align*}
Subtract \eqref{eq:reward} and expand the three entropies to obtain the
equality in \eqref{eq:slack}. Lemma~\ref{lem:kl}(iii), applied to the two
mixture components, makes the right side nonnegative. At $\mu=q$ all
three divergences vanish.
\end{proof}

This proof is global on the closed simplex. It does not extend a
derivative formula to its boundary without justification. In particular,
the actual predictive law of a short or constant input can have zeros;
the reference law alone needs full support.

\section{Finite windows, reachable laws, and the Bellman certificate}\label{sec:window}

For a vector $v\in\R^{64}$ let
$(B_cv)_i=dv_i+(1-d)v_{\sigma_c(i)}$. Define $A_x$ by
\[
 A_\varnothing=I,\qquad A_{cx}=A_xB_c.
\]
Thus, if $x=x_1\cdots x_k$, then
$A_x=B_{x_k}\cdots B_{x_1}$. In evaluating $A_xv$, the operation
$B_{x_1}$ is applied first. This is the reverse composition order to the
forward action on probability laws.

\begin{lemma}[Adjoint and positivity]\label{lem:adjoint}
$\ip{F_x\mu}{v}=\ip{\mu}{A_xv}$.
Moreover $A_x\one=\one$, and $v\le v'$ coordinatewise implies
$A_xv\le A_xv'$ coordinatewise.
\end{lemma}
\begin{proof}
The one-step identity follows by reindexing the pushforward term:
$\sum_j((\sigma_c)_{\#}\mu)_jv_j=\sum_i\mu_iv_{\sigma_c(i)}$.
Iterate it in the stated order. Each $B_c$ is a nonnegative weighted
average with weights summing to one, so it preserves constants and order;
its compositions do as well.
\end{proof}

Fix a positive window length $m$. Let $s(x)\in\bits^m$ be the first
$m$ bits of $x$, padded on the right with zeros if $|x|<m$. Let
$s^+(s,c)$ be the first $m$ bits of $cs$. Define
\[
 \mathbf u=(1/64,\ldots,1/64),\qquad q_s=F_s\mathbf u.
\]
Every coordinate of $q_s$ is at least $d^m/64>0$.

\begin{lemma}[Every actual law is represented]\label{lem:reachable}
For each finite word $x$, there is a law $\mu$ such that
$p_x=F_{s(x)}\mu$. Also $s(cx)=s^+(s(x),c)$.
\end{lemma}
\begin{proof}
If $|x|\ge m$, write $x=s(x)z$ and take $\mu=p_z$. If $|x|=k<m$,
take $\mu=\delta_0$ and note that
$F_0\delta_0=d\delta_0+(1-d)\delta_{\sigma_0(0)}=\delta_0$.
Thus $F_{x0^{m-k}}\delta_0=F_x\delta_0=p_x$. The update identity follows
by checking truncation and padding directly in the two length cases.
\end{proof}

\begin{definition}[Finite certificate]\label{def:certificate}
A potential $h:\bits^m\to\R$ and number $\theta$ form a certificate
for $(w,d,m)$ if, for every $s\in\bits^m$, $c\in\bits$, and $i\in\U$,
\begin{equation}\label{eq:bellman}
 (A_sg_c(q_s))_i+h(s^+(s,c))-h(s)\le\theta.
\end{equation}
\end{definition}

\begin{lemma}[From coordinates to arbitrary suffixes]\label{lem:localcert}
If \eqref{eq:bellman} holds, then for every finite word $x$ and bit $c$,
\[
 r_c(p_x)+h(s(cx))-h(s(x))\le\theta.
\]
\end{lemma}
\begin{proof}
By Lemma~\ref{lem:reachable}, $p_x=F_s\mu$ for $s=s(x)$ and a law $\mu$.
Lemmas~\ref{lem:slack} and~\ref{lem:adjoint} imply
$r_c(p_x)\le\ip{\mu}{A_sg_c(q_s)}$.
Multiply \eqref{eq:bellman} by $\mu_i\ge0$ and sum over $i$; its potential
and constant terms are unchanged because $\sum_i\mu_i=1$. Use the
window update identity.
\end{proof}

\begin{proposition}[Uniform score bound]\label{prop:score}
If \eqref{eq:bellman} holds and $0\le h(s)\le B$ for every state, then
\begin{equation}\label{eq:telescoping}
 S_d(x)+h(s(x))\le |x|\theta+h(0^m),\qquad
 S_d(x)\le |x|\theta+B.
\end{equation}
\end{proposition}
\begin{proof}
The first assertion follows by induction on word length. It is equality
for the empty word. In the induction step add Lemma~\ref{lem:increment}
and Lemma~\ref{lem:localcert}, and apply the induction hypothesis to the
suffix. The second assertion follows from
$h(0^m)-h(s(x))\le B$. Equivalently, one can prepend the bits from
right to left and cancel the consecutive potential terms. The same proof
works for words shorter than $m$; no initial window is discarded.
\end{proof}

\section{Deletion composition and a general converse}\label{sec:converse}

Rahmati and Duman~\cite{RD} proved the capacity inequality
$C(\lambda d_0+1-\lambda)\le\lambda C(d_0)$ for
$0\le\lambda\le1$. The normalized extension used here has the same
form. We prove the underlying channel composition directly and retain
the finite-block correction by using one common auxiliary output law.

\begin{lemma}[Composition of deletion channels]\label{lem:composition}
Fix $0<d_0<1$ and $d_0\le d<1$, and put
$\alpha=(d-d_0)/(1-d_0)$. A deletion channel with probability $\alpha$,
followed independently by a deletion channel with probability $d_0$,
equals $W_d$. The intermediate word $Z$ has
$\E[|Z|\mid x]=(1-\alpha)|x|$.
\end{lemma}
\begin{proof}
$0\le\alpha<1$ and $(1-\alpha)(1-d_0)=1-d$.
For each original coordinate independently draw a first retention bit
with probability $1-\alpha$ and a second with probability $1-d_0$.
Only the second bits attached to first-stage survivors are used. This
constructs the same law as assigning independent second retention bits
to the intermediate word. The products of the two bits are independent
across original coordinates, with retention probability $1-d$. Their
subsequence therefore has kernel \eqref{eq:channel}. The mean-length
identity is the sum of the $|x|$ first-stage retention probabilities.
\end{proof}

\begin{theorem}[General finite-certificate converse]\label{thm:generic}
Let $0<d_0<1$, let $w$ satisfy \eqref{eq:kraft}, and suppose
\eqref{eq:bellman} holds at $d_0$ with $0\le h\le B$ and $\theta\ge0$.
Then for $d_0\le d<1$, $n\ge1$, and every input word $x\in\bits^n$,
\begin{equation}\label{eq:divall}
 \Div(W_d(\cdot\mid x)\Vert Q_n)
 \le \frac{1-d}{1-d_0}n\theta+B+\log(n+1).
\end{equation}
Consequently every input distribution and every block code satisfy
\begin{align}
 \frac{I(X;Y)}{n\log2}
 &\le\frac{1-d}{1-d_0}\frac{\theta}{\log2}
           +\frac{B/\log2+\log_2(n+1)}{n},\label{eq:general-info}\\
 (1-P_e)R_n
 &\le\frac{1-d}{1-d_0}\frac{\theta}{\log2}
           +\frac{B/\log2+\log_2(n+1)+1}{n}.\label{eq:general-code}
\end{align}
\end{theorem}
\begin{proof}
First consider any intermediate word $z$ of length $k\le n$. Its
second-stage output law naturally embeds in $\Y_n$. Embedding preserves
its entropy and expected cost: it only adds zero-mass coordinates.
Lemma~\ref{lem:partition} and Proposition~\ref{prop:score}, at $d_0$, give
\[
 \Div(W_{d_0}(\cdot\mid z)\Vert Q_n)
 \le k\theta+B+\log(n+1).
\]
Crucially, the same positive law $Q_n$ is used for every value of $k$.
By Lemma~\ref{lem:composition}, $W_d(\cdot\mid x)$ is the mixture of
these laws over $Z\sim W_\alpha(\cdot\mid x)$. Apply convexity in the
first argument of divergence, Lemma~\ref{lem:kl}(iii), and then the
mean-length identity. This proves \eqref{eq:divall}.

Apply \eqref{eq:dualidentity}, discard the nonnegative output divergence,
and use \eqref{eq:divall} to obtain \eqref{eq:general-info}.
For the decoding assertion we use the elementary form of Fano's
inequality~\cite[Ch.~2]{CT}, whose short proof follows. Let $U$ be the
uniform message,
$\widehat U=g(Y)$, and $E=\one_{\{U\ne\widehat U\}}$. Since $E$ is
determined by $(U,Y)$, the conditional entropy chain rule gives
\begin{align*}
 \Ent(U\mid Y)
 &=\Ent(E\mid Y)+\Ent(U\mid E,Y)\\
 &\le\log2+P_e\log M.
\end{align*}
Indeed, when $E=0$ the message is determined by $Y$, and when $E=1$
its conditional entropy is at most $\log M$. This argument also holds
for $M=1$. Therefore
\[
 (1-P_e)\log M\le I(U;Y)+\log2.
\]
Equation~\eqref{eq:dualidentity} applied to the message channel bounds
$I(U;Y)$ by the right side of \eqref{eq:divall}, since each message
selects one input word. Divide by $n\log2$.
\end{proof}

\begin{corollary}[Operational and sequence consequences]\label{cor:operational}
Under the hypotheses of Theorem~\ref{thm:generic}, set
$a_d=(1-d)\theta/((1-d_0)\log2)$. Then $C(d)\le a_d$.
Every sequence of reliable codes with lengths tending to infinity has,
for every $\varepsilon>0$, eventual rate strictly below $a_d+\varepsilon$.
\end{corollary}
\begin{proof}
Let $b_n=(B/\log2+\log_2(n+1)+1)/n$. The elementary limit
$\log(n+1)/n\to0$ gives $b_{n_j}\to0$. Since $P_{e,j}\to0$,
eventually $P_{e,j}<1$, and \eqref{eq:general-code} gives
\[
 R_{n_j}\le\frac{a_d+b_{n_j}}{1-P_{e,j}}\longrightarrow a_d
 \quad\text{on the right-hand side.}
\]
This proves the eventual assertion without assuming that rates converge
or are bounded in advance. For any achievable rate apply it to the
all-length sequence, then take the rate limit. The achievable set is
nonempty and bounded above by $a_d$, so its supremum is at most $a_d$.
\end{proof}

\begin{remark}[What kind of converse is proved]
The finite-block inequality gives, for $R_n>0$,
$P_e\ge1-(a_d+b_n)/R_n$. In particular, if $R_{n_j}\to R>a_d$,
then $\liminf_jP_{e,j}\ge1-a_d/R>0$. This elementary consequence
does not say that the error tends to one; a strong converse is not
asserted.
\end{remark}

\section{The exact numerical certificate}\label{sec:integer}

It remains to supply a potential and code lengths satisfying the finite
inequality. We describe the chosen values and reduce the required
inequalities to integer comparisons. The reduction accounts explicitly
for the errors in the logarithm bounds.

We now specialize to
\begin{equation}\label{eq:parameters}
 d_0=\frac{13}{20},\quad m=26,\quad
 \theta=\frac7{80}\log2,\quad B=\log2.
\end{equation}
The code lengths are $w_{i,c}=N_{i,c}/D$ with $D=2^{55}$.
All 128 integer numerators are given in Appendix~\ref{app:code}.
The potential is
\begin{equation}\label{eq:potential}
 h(s)=\frac{\log2}{P}H(s),\qquad P=65536,
 \qquad H(s)\in\{0,\ldots,65535\}.
\end{equation}
Thus $0\le h(s)<\log2$. Section~\ref{sec:coverage} specifies the state
indexing, and Appendix~\ref{app:witness} identifies the electronic data.

\subsection{Exact reference numerators and adjoints}
For an integer vector $v\in\N^{64}$ define
\begin{align}
 (M_cv)_j&=\sum_{i:\,\sigma_c(i)=j}v_i,
 & (G_cv)_j&=13v_j+7(M_cv)_j,\label{eq:intforward}\\
 (T_cv)_i&=13v_i+7v_{\sigma_c(i)}.\label{eq:intadjoint}
\end{align}
The fibers of $\sigma_c$ are explicit:
\[
 (M_cv)_j=\begin{cases}
 v_{\lfloor j/2\rfloor}+v_{\lfloor j/2\rfloor+32},&j\bmod2=c,\\
 0,&j\bmod2\ne c.
 \end{cases}
\]
Starting with $v=\one$, process the bits of $s=s_1\cdots s_m$ from
$s_m$ to $s_1$, replacing $v$ by $G_{s_j}v$. The resulting vector
$q^{\mathrm{int}}$ satisfies
\begin{equation}\label{eq:refnumerator}
 q_s=\frac{q^{\mathrm{int}}}{64\cdot20^m},\qquad
 q^{\mathrm{int}}_i\ge13^m>0,\qquad
 \sum_iq^{\mathrm{int}}_i=64\cdot20^m.
\end{equation}
These identities follow inductively from \eqref{eq:intforward}, since
pushforward preserves total mass. For an action $c$, write
$t=G_cq^{\mathrm{int}}$ and
$v_j=q^{\mathrm{int}}_j+q^{\mathrm{int}}_{j+32}$ for $0\le j<32$.
The integer adjoint is
$T_s=T_{s_m}\cdots T_{s_1}=20^mA_s$.

\subsection{Eliminating every real logarithm from the check}
Let $S=2^{40}$. The certified logarithm functions used by the checker
return integers $\ell(a),u(a)$ satisfying
\begin{equation}\label{eq:logenclosure}
 \frac{\ell(a)}{S}\le\log a\le\frac{u(a)}{S}
 \qquad(0<a<2^{144}).
\end{equation}
Their complete finite description and soundness argument are given in
Appendix~\ref{app:logs}. We also use the exact constants
\begin{equation}\label{eq:logconstants}
 U_2=762123384787,\qquad L_{20}=3293842468468,
\end{equation}
for which
$(U_2-3)/S\le\log2\le U_2/S$ and $L_{20}/S\le\log20$.
Set $E=2^{112}$ and $Q=20DS$. For each state and action define
\begin{align}
 a_i={}&E+7N_{i,c}U_2\notag\\
 &+D\left(13u(t_i)+7u(t_{\sigma_c(i)})
       -13\ell(q_i^{\mathrm{int}})-7\ell(v_{i\bmod32})-20L_{20}\right).
                                                        \label{eq:aint}
\end{align}
Although this expression is written with integer subtraction, its value
is nonnegative. The exact arithmetic bounds give $0\le a_i<2^{113}$;
Appendix~\ref{app:implementation} explains their role in the computation.

\begin{lemma}[Integer majorant]\label{lem:intmajorant}
For every state, action, and coordinate,
$g_c(q_s)_i\le(a_i-E)/Q$.
\end{lemma}
\begin{proof}
Substitute \eqref{eq:refnumerator} into \eqref{eq:g}. The four
occurrences of the reference denominator cancel, and the extra factor
20 in $F_cq_s$ contributes $-\log20$. Hence
\begin{align*}
 20g_c(q_s)_i={}&7w_{i,c}\log2+13\log t_i+7\log t_{\sigma_c(i)}\\
 &-13\log q_i^{\mathrm{int}}-7\log v_{i\bmod32}-20\log20.
\end{align*}
Every logarithm argument is positive. Use upper enclosures for the
positive coefficients and lower enclosures for the negative coefficients.
The fixed $w_{i,c}$ are nonnegative, so their $\log2$ factor uses $U_2$.
Multiplication by $DS$ gives the claim. The argument bounds
$q^{\mathrm{int}},t,v<2^{144}$, for example by their total masses
$64\cdot20^{m+1}$ and $m\le28$; thus all logarithm enclosures apply.
\end{proof}

For current and successor potential values $H,H'$, put
\begin{equation}\label{eq:delta}
 \Delta=80(H'-H)-7P,\qquad
 b(H,H')=E-2^{37}(\Delta U_2+3\cdot2^{23}).
\end{equation}
The exhaustive integer predicate is
\begin{equation}\label{eq:integercheck}
 (T_sa)_i\le20^m b(H(s),H(s^+(s,c)))
 \qquad(s\in\bits^m,\ c\in\bits,\ i\in\U).
\end{equation}

\begin{lemma}[Acceptance implies the Bellman inequality]\label{lem:intsound}
If \eqref{eq:integercheck} holds, then \eqref{eq:bellman} holds with
the parameters \eqref{eq:parameters} and potential \eqref{eq:potential}.
\end{lemma}
\begin{proof}
Since $0\le H,H'<P$, one has $\Delta\ge-2^{23}$. Therefore
\begin{equation}\label{eq:signallowance}
 \Delta\log2\le\frac{\Delta U_2+3\cdot2^{23}}{S}.
\end{equation}
For $\Delta\ge0$ use the upper bound on $\log2$. For $\Delta<0$,
use the lower bound, multiply by the negative $\Delta$, and use
$-3\Delta\le3\cdot2^{23}$. Thus the uniform allowance is valid with
either sign of the potential increment.

By Lemmas~\ref{lem:adjoint} and~\ref{lem:intmajorant},
\[
 (A_sg_c(q_s))_i\le\frac{(T_sa)_i-20^mE}{20^mQ}
 \le-\frac{2^{37}}{Q}(\Delta U_2+3\cdot2^{23}).
\]
The exact scale identity
\[
 \frac{2^{37}}{Q}=\frac{1}{80PS}
\]
and \eqref{eq:signallowance} make this at most
$-\Delta\log2/(80P)$. Finally
\[
 \frac{\Delta\log2}{80P}
 =h(s^+(s,c))-h(s)-\frac7{80}\log2.
\]
Rearrangement proves the desired Bellman inequality.
\end{proof}

\section{The certificate and completion of the proof}\label{sec:coverage}

Encode the window by
$s=\sum_{j=1}^{26}2^{j-1}s_j\in\{0,\ldots,2^{26}-1\}$, and encode
an action by $r=2s+c$. Its successor window is
$s^+=(2s+c)\bmod2^{26}$. Thus the action range is exactly
$0\le r<2^{27}$. No symmetry reduction or restriction to typical inputs
is used in the coverage argument.

There are $2^{26}$ states, two choices of the prepended bit, and 64
coordinates. Thus the certificate consists of
$2^{26}\cdot2\cdot64=2^{33}$ scalar inequalities. The following proposition
is the computer-assisted step of the proof.

\begin{proposition}[Concrete finite witness]\label{prop:concrete}
The code table of Appendix~\ref{app:code} satisfies \eqref{eq:kraft}.
The potential $H$ supplied in the electronic supplement has values in
$\{0,\ldots,65535\}$, and
all inequalities \eqref{eq:integercheck} hold for that potential at $m=26$.
\end{proposition}
\begin{proof}[Computer-assisted proof]
The 64 Kraft inequalities follow from exact rational bounds on the
exponential function, as detailed in Appendix~\ref{app:kraft}. The values
of $H$ are integers between zero and $65535$ by their representation as
16-bit digits.

For this fixed $H$, the Lean certificate proves every instance of
\eqref{eq:integercheck}. The checks cover every state, both input bits,
and every coordinate. They use integer arithmetic and the logarithm
bounds established in Appendix~\ref{app:logs}. To make the computation
manageable, the implementation combines many scalar operations into
operations on larger integers. Appendix~\ref{app:implementation} explains
why this representation preserves each inequality and uses the same
potential at every occurrence of a state.

The resulting finite proof establishes precisely the inequalities stated
here. Appendix~\ref{sec:formal} identifies the formal theorems and their
verification record; Appendix~\ref{app:algorithm} gives the scalar
calculation in ordinary mathematical notation.
\end{proof}

\begin{proof}[Proof of Theorem~\ref{thm:main}]
Proposition~\ref{prop:concrete} and Lemma~\ref{lem:intsound} establish
the finite Bellman certificate. Equation~\eqref{eq:potential} gives
$0\le h\le\log2$. Apply Theorem~\ref{thm:generic} with
\eqref{eq:parameters}. The coefficient simplifies exactly:
\[
 \frac{1-d}{1-13/20}\frac{(7/80)\log2}{\log2}
 =\frac{1-d}{4}.
\]
The two finite-block correction terms become those in
\eqref{eq:main-info} and~\eqref{eq:main-code}.
Corollary~\ref{cor:operational} gives both the operational capacity
bound and the arbitrary-length-sequence conclusion.
\end{proof}

\section{Consequences, limitations, and further work}\label{sec:discussion}

The certificate produces a bound for every deterministic input word
before any input distribution is chosen. It consequently allows arbitrary
input dependence and nonstationary encoders. The only randomness needed
for the channel model is independent deletion. No infinite-output test
measure, stationary maximizing process, interchange of an extremum with
a limit, or assumption of rate convergence enters the finite-block proof.

At the endpoint $d_0=13/20$, the rate coefficient is $7/80$ bits per
input bit. The factor $1/(1-d_0)=20/7$ gives the normalized coefficient
$1/4$. The boundary correction is one bit for the potential and at most
$\log_2(n+1)$ bits for the output-length partition bound. The decoding
statement adds one further bit. These corrections are uniform in the
window length once the potential range is fixed.

The size of the witness is a practical limitation. There are $2^{26}$
potential values and $2^{33}$ scalar inequalities, even though the
predictive law itself has only 64 coordinates. The proof is finite and
fully specified, but its independent reproduction is computationally
substantial. A smaller potential representation or a sharper symbolic
bound could reduce the cost without changing the channel argument.

A natural next question is whether a compact analytical potential or a
richer predictive law can improve the coefficient. This paper establishes
exactly $1/4$; it makes no claim that this coefficient is optimal.

\section*{Data and code availability}
The electronic certificate, Lean proof, and reproduction instructions are
available at \url{https://github.com/factoreminv/bdc}.
Appendix~\ref{app:witness} identifies the exact potential used in the proof.

\appendix
\input{appendices.tex}

\end{document}

%% file: appendices.tex
\section{Exact code lengths}\label{app:code}

Table~\ref{tab:code} gives the complete integer code table. Its columns
are numerators with the common denominator
$D=36028797018963968=2^{55}$, not rounded decimal approximations to
the lengths. Row $i$ uses the context encoding \eqref{eq:context};
column $c$ is the symbol being prepended. The table is extracted from
\path{PackedCode.lean} and independently compared as rational numbers
with the entries \path{w_i_c} in \path{FixedCode.lean} by the manuscript
support script. Its extraction record is supplied in
\path{evidence/paper-data-checks.json}.

\input{code_table.tex}

\section{Rational proof of the Kraft inequalities}\label{app:kraft}

For rational $x$ and an integer $n\ge1$, define
\[
 P_n(x)=\sum_{j=0}^{n-1}\frac{x^j}{j!},\qquad
 R_n(x)=\frac{|x|^n(n+1)}{n!\,n},\qquad
 E_n^{\pm}(x)=P_n(x)\pm R_n(x).
\]
For $|x|\le1$, the exponential series yields
\begin{equation}\label{eq:exptaylor}
 E_n^-(x)\le e^x\le E_n^+(x).
\end{equation}
Indeed, after the first omitted term, every successive term has
absolute ratio at most $1/(n+1)$. The absolute tail is bounded by
\[
 \frac{|x|^n}{n!}\sum_{k=0}^{\infty}\frac1{(n+1)^k}
 =R_n(x).
\]

Set
\[
 \lambda_-=
 \frac{693147180559945309417232121458}{10^{30}},\qquad
 \lambda_+=
 \frac{693147180559945309417232121459}{10^{30}}.
\]
The exact rational comparisons $E_{32}^+(\lambda_-)\le2$ and
$2\le E_{32}^-(\lambda_+)$, followed by monotonicity of the
exponential, prove $\lambda_-\le\log2\le\lambda_+$.

Every code length in Table~\ref{tab:code} lies in $[0,3]$. For
$w=w_{i,c}$ let $x=-w\lambda_-/4$, so $|x|\le1$. Then
\begin{equation}\label{eq:kraftfactor}
 2^{-w}\le e^{-w\lambda_-}
 =\bigl(e^{-w\lambda_-/4}\bigr)^4
 \le\bigl(E_{32}^+(-w\lambda_-/4)\bigr)^4.
\end{equation}
The exponential enclosure is positive since it bounds $e^x>0$ from
above. The formal data supply rational $u_{i,c}$ satisfying
\[
 E_{32}^+(-w_{i,c}\lambda_-/4)\le u_{i,c},\qquad
 u_{i,0}^4+u_{i,1}^4\le1.
\]
Combining with \eqref{eq:kraftfactor} proves the Kraft row. These are
finite rational comparisons; the stored $u_{i,c}$ are conveniences for
the proof, not additional real-valued assumptions. They can be read
verbatim in \path{FixedCode.lean}. The paper support script independently
checks all 128 exponential upper comparisons and all 64 fourth-power
row comparisons using rational arithmetic. It also directly checks
\[
 \bigl(E_{32}^+(-w_{i,0}\lambda_-/4)\bigr)^4+
 \bigl(E_{32}^+(-w_{i,1}\lambda_-/4)\bigr)^4\le1
\]
for every row, so the printed code table and the stated Taylor formula
suffice to reproduce the Kraft calculation without the auxiliary
$u_{i,c}$. This auxiliary check is
reported as a data check, not as a replacement for the Lean proofs.
\lean{ExponentialBounds.lean; FixedCode.lean; PackedCode.lean}

\section{The logarithm enclosure routine}\label{app:logs}

This appendix specifies the finite logarithm data used in
\eqref{eq:aint}. It distinguishes the analytic enclosure, the integer
grid, and the optimized affine lookup.

\subsection{Positive centers and directed rounding}
For $x,c>0$ and $l\le\log c\le u$, the inequality
$\log t\le t-1$ gives
\begin{equation}\label{eq:center}
 l+\frac{x-c}{x}\le\log x\le u+\frac{x-c}{c}.
\end{equation}
For the lower inequality apply the elementary bound to $c/x$ and
negate; for the upper one apply it to $x/c$. Thus the residual can
have either sign without changing the formulas.

The base table uses the 129 rational centers $c_j=1+j/128$,
$0\le j\le128$, and integer endpoints $L_j,U_j$ such that
$L_j/S\le\log c_j\le U_j/S$. Each endpoint is connected to a
rational exponential comparison using \eqref{eq:exptaylor}, and those
comparisons are proved in the formal base table. A logarithm enclosure
is not assumed merely because an external program computed it.

For any $a\ge1$, exponent $k\ge0$, and center $j$, put
$\rho=128a-(128+j)2^k$. Equation~\eqref{eq:center}, applied at
$x=a/2^k$, gives the integer bounds
\begin{align}
 \ell_0(a;k,j)&=k(U_2-3)+L_j+
                  \left\lfloor\frac{S\rho}{128a}\right\rfloor,
                                                        \label{eq:base-lower}\\
 u_0(a;k,j)&=kU_2+U_j+
                  \left\lceil\frac{S\rho}{(128+j)2^k}\right\rceil.
                                                        \label{eq:base-upper}
\end{align}
They satisfy $\ell_0/S\le\log a\le u_0/S$. The floor and ceiling
here are mathematical signed floor and ceiling, including for negative
$\rho$. The denominators are strictly positive. In particular, a
center-selection algorithm need not be accurate for soundness: every
center and exponent obey these inequalities. Selection affects only
the width of the enclosure.

\subsection{Normalized grid}
For an integer $0<a<2^{144}$ choose the unique integer $K\ge0$ with
$2^K\le a<2^{K+1}$ and let
\[
 M=\left\lfloor\frac{a\,65536}{2^K}\right\rfloor.
\]
Then $65536\le M<131072$ and
\begin{equation}\label{eq:normalizer}
 2^K\frac{M}{65536}\le a
 <2^K\frac{M+1}{65536}.
\end{equation}
The normalized grid has 65,537 endpoints, including $131072/65536=2$.
For its point $x_j=(65536+j)/65536$, let $L_j^{\rm f},U_j^{\rm f}$
be the lower and upper integer bounds obtained from
\eqref{eq:base-lower}--\eqref{eq:base-upper} at numerator $65536+j$,
exponent 16, and center index $\min(\lfloor j/512\rfloor,128)$,
subtracting respectively $16U_2$ and $16(U_2-3)$.
The formal implementation certifies its stored fine-grid lookup data
against this specification. By monotonicity of the logarithm,
\begin{align}
 \ell_{\rm n}(M,K)&=K(U_2-3)+L_{M-65536}^{\rm f},\label{eq:normlower}\\
 u_{\rm n}(M,K)&=KU_2+U_{M+1-65536}^{\rm f}\label{eq:normupper}
\end{align}
are lower and upper integer enclosures for $S\log a$.
The upper endpoint uses $M+1$, not $M$; dropping that distinction
would invalidate the bound for truncated mantissas.

\subsection{Thirty-two affine pieces}
Table~\ref{tab:logs} gives two arrays of packed integers $C_j,V_j$,
$0\le j<32$. Put
\[
 j=\left\lfloor\frac{M-65536}{2048}\right\rfloor,\qquad t=M\bmod2048,
 \qquad Z_2=14058695031755056497480129050576.
\]
Define the packed pair
\begin{equation}\label{eq:logpair}
 P_{M,K}=C_j+tV_j+2^{64}\left\lfloor V_j/2^{64}\right\rfloor+KZ_2
\end{equation}
and the integer endpoints
\begin{align}
 \ell(a)&=(P_{M,K}\bmod2^{64})-64,\label{eq:affinelower}\\
 u(a)&=\left(\left\lfloor P_{M,K}/2^{64}\right\rfloor\bmod2^{64}\right)-64.
                                                        \label{eq:affineupper}
\end{align}
These equations and Table~\ref{tab:logs} fully specify the scalar
enclosures used by the optimized checker. The bias of 64 keeps the
packed fields nonnegative. Its removal in the raw-gradient expression
is exact because $13+7-13-7=0$.

For all 65,536 normalized mantissas, the formal finite checks establish
that the affine lower value at $K=0$ is no larger than
$\ell_{\rm n}(M,0)$ and the affine upper value is no smaller than
$u_{\rm n}(M,0)$. They also bound each initial packed field below
$2^{40}$. The exponent extension adds the two appropriate integer
endpoints for $\log2$ without a carry between the fields for $K\le255$.
Equations~\eqref{eq:normalizer}--\eqref{eq:normupper} then prove
\eqref{eq:logenclosure}. Our argument range gives $K\le143$, within
the proved extension range. The comparison is finite over mantissas
but universal over all positive integer arguments in the stated range.
\lean{LogArithmetic.lean; FixedLogTable.lean; FastLogTable.lean; FastLog.lean}
\lean{FineLogSpecification.lean; FineLogTable.lean; NormalizedLog.lean}
\lean{LinearLogData.lean; LinearLogTable.lean; BatchLogSound.lean}

\input{log_table.tex}

\section{A scalar specification of the exhaustive check}\label{app:algorithm}

The following procedure is a readable specification of
\eqref{eq:integercheck}. Its variables are mathematical integers of
unbounded size; logarithm endpoints are the integer functions of
Appendix~\ref{app:logs}. It is not the packed implementation used to
accelerate the Lean proof.

\begin{enumerate}
\item Read the exact potential $H[0],\ldots,H[2^{26}-1]$ and the code
numerators $N_{i,c}$. Require $0\le H[s]<65536$.
\item For each integer $s$ from 0 through $2^{26}-1$, define
$s_j=\lfloor s/2^{j-1}\rfloor\bmod2$ for $1\le j\le26$.
Start $q_i^{\mathrm{int}}=1$ and apply $G_{s_j}$ for $j=26,25,\ldots,1$.
\item For each $c\in\{0,1\}$, form $t=G_cq^{\mathrm{int}}$ and
$v_j=q_j^{\mathrm{int}}+q_{j+32}^{\mathrm{int}}$.
Use \eqref{eq:aint} to obtain $a$.
\item Set $z=a$. For $j=1,2,\ldots,26$, replace $z$ simultaneously by
$T_{s_j}z$. ``Simultaneously'' means that all new coordinates use the
old vector; an in-place coordinate loop without a copy is not this map.
\item Set $s'=(2s+c)\bmod2^{26}$, compute $b(H[s],H[s'])$ from
\eqref{eq:delta}, and require $z_i\le20^{26}b(H[s],H[s'])$ for all
$i=0,\ldots,63$. Reject if any comparison fails.
\end{enumerate}

The grouped proof establishes acceptance of this entire predicate
through representation lemmas. It does not sample states. One full
scalar implementation could independently check the same witness, but
the manuscript support script does not claim to have rerun this
$2^{33}$-inequality procedure. Its purpose is to reproduce the paper's
tables and evidence snapshot; the supplied Lean development is the
proof-producing verifier.

\section{Exact arithmetic and storage}\label{app:implementation}

\subsection{Packed arithmetic}
The production proof evaluates batches of 64 actions, with 64 coordinates
per action. Each coordinate occupies a lane of width 256 in a natural
number. For lane values $z_k$, define
$\operatorname{pack}(z)=\sum_k z_k2^{256k}$. Addition and multiplication
by a scalar have their coordinatewise interpretation whenever the
resulting lane values are below $2^{256}$; this follows by uniqueness
of radix expansion. Natural subtraction agrees with coordinatewise
subtraction when each subtracted lane is no larger than its corresponding
minuend. The implementation proves these preconditions at each operation.

For the final comparison the sharper bounds are
\[
 0\le (T_sa)_i<20^m2^{113}<2^{255},\qquad
 0\le20^mb(H,H')<2^{255}\quad(m\le28).
\]
The first follows by induction from
$T_c a\le20\max_i a_i$ and the certified initial bound. For the second,
write $b=b_0+b_1H-b_1H'$ with
\[
 b_0=E+2^{37}(7PU_2-3\cdot2^{23}),\qquad b_1=2^{37}80U_2.
\]
The exact integer inequalities $b_1P\le b_0$ and
$b_0+b_1P<2^{113}$ prove nonnegativity and the required upper bound,
without allowing a truncated subtraction.

For $0\le A,B<2^{255}$, the integer $B+2^{255}-A$ lies strictly
between zero and $2^{256}$; its bit at position 255 is one exactly when
$A\le B$. Applying this fact lane by lane proves that the packed
comparison tests every scalar inequality in \eqref{eq:integercheck}.
It is not an average test, a floating-point tolerance, or a machine-word
overflow heuristic. The kernel works with natural numbers, while the
lane lemmas justify their intended vector interpretation.
\lean{RawBatchGradient.lean; RawBatchBounds.lean; RawBatchSound.lean}
\lean{PackedArrayComparison.lean; PackedMatrixArithmetic.lean}

\subsection{A single potential across all batches}
Store the potential as chunks
\[
 D_b=\sum_{j=0}^{31}H(32b+j)2^{16j},\qquad 0\le b<2^{21}.
\]
Each $D_b<2^{512}$. For a batch $r=64b+k$, $0\le k<64$,
\begin{align*}
 s&=32b+\lfloor k/2\rfloor,\\
 H(s)&=\left\lfloor D_b/2^{16\lfloor k/2\rfloor}\right\rfloor\bmod65536,\\
 H(s^+)&=\left\lfloor
  \bigl(D_{2b\bmod2^{21}}+2^{512}D_{(2b+1)\bmod2^{21}}\bigr)/2^{16k}
                      \right\rfloor\bmod65536.
\end{align*}
The last formula follows by writing the successor index $64b+k$ in
chunks of 32 and then reducing it modulo $2^{26}$. For $k<32$ it reads
the first chunk, and for $k\ge32$ it reads the second. These identities
prove that current and successor values are evaluations of one global
potential, even where a batch or group boundary is crossed.

\subsection{Finite coverage and the concrete proposition}
There are $2^{21}$ batches, divided into $2^{10}=1024$ groups of
$2^{11}=2048$ batches. Every action has a unique decomposition
\[
 r=64(2048g+b)+k,\quad
 0\le g<1024,\quad0\le b<2048,\quad0\le k<64.
\]
Existence and uniqueness follow from successive Euclidean divisions.
Every action then has 64 coordinate obligations, for a total of
$2^{27}\cdot64=2^{33}=8,589,934,592$ scalar inequalities.

\subsection{Identification of the electronic certificate}\label{app:witness}
The source potential file, relative to \path{upper_formal}, is
\begin{center}\small\path{results/small-windows/quantized-m26-howard.npy}.\end{center}
Its SHA-256 is
\begin{center}\footnotesize
\texttt{5e6818890131d589388c063634c3c975ef9c28c59962a932eaa1e38ec0ac3282}.
\end{center}
The array contains $2^{26}$ unsigned 16-bit values, a 128 MiB payload
before its file header. The generated Lean data encode these values in
chunks and lookup trees. The supplementary evidence snapshot identifies
the source files and available round-trip checks. A hash identifies the
byte sequence; it does not prove the inequalities. Logically, the Lean
theorem is about the literal data it imports. Reproducing those data
from the array is an additional provenance check, distinct from kernel
verification of the imported literal witness.

The numerical search that proposed this witness is outside the proof
dependencies. Neither convergence of a policy iteration nor accuracy of
a floating-point eigenvalue calculation is assumed. A proposed potential
is acceptable exactly when the integer checker and its proved soundness
implication establish the required inequalities.

\section{Formal correspondence and verification record}\label{sec:formal}

\subsection{The five closed declarations}
The final file is \path{upper_formal/BDCFormal/Howard26Upper.lean}.
Its declarations are:
\begin{enumerate}
\item \path{BDCFormal.upper_capacity_bound}: the operational capacity
bound for all $13/20\le d<1$;
\item \path{BDCFormal.upper_finite_block_bound}: \eqref{eq:main-code}
for every positive length and every block code;
\item \path{BDCFormal.upper_information_bound}: \eqref{eq:main-info}
for every law on the entire length-$n$ input alphabet;
\item \path{BDCFormal.upper_sequence_bound}: the limit-rate converse
for arbitrary input lengths tending to infinity;
\item \path{BDCFormal.upper_eventual_rate_bound}:
\eqref{eq:main-eventual}, without a rate-convergence hypothesis.
\end{enumerate}
The final capacity declaration has the following form:
\begin{Verbatim}[fontsize=\small,commandchars=\\\{\}]
theorem upper_capacity_bound :
    \ensuremath{\forall} d (hd : 13/20 \ensuremath{\le} d) (hd1 : d < 1),
    operationalCapacity d (by linarith) hd1.le \ensuremath{\le} (1-d)/4 := by
  intro d hd hd1
  exact RawBatch.capacity_of_groups
    Howard26.groups Howard26.exhaustive d hd hd1
\end{Verbatim}
Its hypotheses are the deletion-parameter inequalities. The finite
certificate is supplied by a proved term, not left as a hypothesis of
the final theorem. The full body of the actual declaration, rather than
this typeset excerpt, is the formal artifact. Appendix~\ref{app:map}
maps the main mathematical steps to source modules.

\subsection{Verification record}
Verification completed on September 11, 2026. All 1,024 certificate groups
compiled, and the analytic proof and 256 certificate components passed
fresh replay. The final theorems were then checked over those verified
imports. This was a modular verification, with each component replayed
separately. The completed report and logs are preserved in
\path{upper_formal/evidence/completed-audit/}.

The audit checks the theorem statements, their dependencies, the fixed
channel definitions, and the identity of the source and compiled files.
It also tests that deliberately introduced unproved assumptions are
rejected. Commands for reproducing these checks are given in
Appendix~\ref{app:reproduction}; the detailed counts and results are in
\path{verification.json} in the directory above.

\subsection{Foundations and limits of the claim}
The development uses Lean~4~\cite{Lean4} and the Mathlib
library~\cite{Mathlib}. It is pinned to Lean \texttt{4.34.0-rc2}
and Mathlib commit
\begin{center}\small
\texttt{d6893048e0d784c43f3cf098b61299b3a4b4aed0}.
\end{center}
The audit permits the usual foundational dependencies
\mbox{\path{propext}}, \mbox{\path{Classical.choice}}, and \mbox{\path{Quot.sound}}; it is
designed to reject additional axioms and admission-based proofs.
Large finite predicates use kernel evaluation and proved representation
identities. External numerical optimization is not accepted as a proof
of a real inequality. The completed audit checked this policy over
the final dependency closure, including deliberate negative controls.

Formal verification establishes propositions about the specified
mathematical definitions, subject to the proof assistant's foundations
and implementation. Matching those definitions to the intended channel
is a separate semantic responsibility; the explicit model
\eqref{eq:channel}, the code definition, and the finite-block theorem
make that responsibility inspectable. The formal guarantee is relative to these explicit definitions
and the stated foundations.

\section{Correspondence between mathematics and Lean}\label{app:map}

All paths in this appendix are relative to
\path{upper_formal/BDCFormal}. A module name identifies where the
relevant definitions and proofs are developed; the exact final theorem
names are in Appendix~\ref{sec:formal}.

\paragraph{Finite probability and the actual channel.}
\path{FiniteProbability.lean}, \path{FiniteEntropy.lean}, and
\path{DeletionChannel.lean} define finite laws, entropy, retention
masks, and the deletion channel. \path{Operational.lean} defines
messages, encoders, decoders, average error, achievable rates, and
operational capacity. These are the definitions summarized in
Section~\ref{sec:model}.

\paragraph{Entropy and divergence.}
\path{LogSum.lean}, \path{DivergenceMixture.lean}, and
\path{EntropyDataProcessing.lean} supply the support-sensitive
log-sum, mixture, and information-contraction arguments of
Section~\ref{sec:entropy}. \path{DecoderConverse.lean} connects
uniform-message decoding error to divergence bounds.

\paragraph{Backward cost and predictive recursion.}
\path{BackwardCode.lean} builds the finite positive test laws and their
partition bounds. \path{DeletionRecursion.lean} and
\path{RewardRecursion.lean} relate prepending an actual input bit to
the predictive filter and score increment. \path{FilterTangent.lean}
proves the exact relative-entropy slack and global majorant.

\paragraph{Window and word certificates.}
\path{WindowReachability.lean} proves reachability with zero padding.
\path{FiniteWindowConverse.lean} defines
\path{FiniteWindowCertificate}, its coordinate interpretation, and
the reduction to a word certificate.
\path{WordPotentialConverse.lean} telescopes that certificate.

\paragraph{Thinning and operational consequences.}
\path{ChannelThinning.lean} proves the actual channel composition.
\path{ThinningConverse.lean}, \path{FiniteBlockUpper.lean}, and
\path{InformationBlockUpper.lean} establish the finite output-law,
decoding, and information conclusions. \path{OperationalConverse.lean},
\path{SequenceConverse.lean}, and
\path{SequenceUniformConverse.lean} provide the limiting results.

\paragraph{Exact real-to-integer bridge.}
\path{PackedCode.lean} connects the packed numerators to the rational
code table through the identity \path{PackedCode.rational_eq}. \path{IntegerFilter.lean}, \path{RawGradient.lean}, and
\path{RawBellman.lean} connect integer references, logarithm
enclosures, and the signed potential correction to the real Bellman
inequality. \path{RawBatchSound.lean} assembles the packed arithmetic
lemmas. \path{RawBatchCapacity.lean} and
\path{FiniteBlockGroups.lean} connect accepted global groups to the
information-theoretic conclusions.

\paragraph{Global data and optimized evaluators.}
\path{ChunkPotential.lean} and \path{ChunkGroups.lean} prove the
potential indexing and group coverage. The optimized evaluators are
connected to the original checker in
\path{FastCheck64Sound.lean}, \path{LiteralCheck64Sound.lean}, and
their metadata and lookup support modules. The exhaustive data and
group proofs are under \path{Howard26Data/} and
\path{Howard26Checks/}. Their final assembly is
\path{Howard26Upper.lean}.

\section{Reproduction and document provenance}\label{app:reproduction}

In the computational release, \path{upper_formal} contains the formal
project and \path{quarter_bound_paper} contains the manuscript sources,
exact tables, evidence, and building utilities. The release includes every
local proof dependency and the concrete literal data; pinned external Lean
and Mathlib dependencies are obtained during setup. The root
\path{README.md} distinguishes the completed audit evidence from the
working directory of a fresh reproduction.

After installing the pinned Lean toolchain, initialize dependencies from
\path{upper_formal} with
\begin{Verbatim}[fontsize=\small]
lake exe cache get
\end{Verbatim}
From the formal project directory, the assembly build is
\begin{Verbatim}[fontsize=\small]
lake build BDCFormal.Howard26Upper
\end{Verbatim}
To reproduce the explicit group-check registration and then the strict
audit in an independent checkout, use
\begin{Verbatim}[fontsize=\small]
python3 scripts/run_howard_certificate.py --phase checks \
  --workers 4 --min-free-gib 3
python3 scripts/verify_upper.py --modular-fresh --require-upper
\end{Verbatim}
The toolchain and Mathlib revision must match the pinned versions.
These commands are resource-intensive. A fresh checkout has no compiled
certificate cache, so compilation and independent replay must both run. The manuscript
build does not run them or alter the formal source tree.

To rebuild this paper from its recorded snapshot, run
\begin{Verbatim}[fontsize=\small]
python3 build_paper.py
\end{Verbatim}
To deliberately capture a new snapshot and regenerate the exact tables,
run \path{prepare_paper.py} before building. The snapshot routine
checks the expected code-table and Kraft relationships, records file
hashes, and requires a completed verification report. The build
script runs LaTeX sufficiently many times to resolve cross-references
and rejects unresolved references or overfull boxes.

\paragraph{Scope of the source comparison.}
The paper is a mathematical exposition of the implemented chain, not a
machine-checked translation of the prose itself. Its elementary
log-sum and series arguments are written out for a human reader;
the Lean implementation may use equivalent finite-law identities or
library lemmas. The finite-block theorem, operational definitions,
coefficient, state encoding, and integer predicate are matched
explicitly. The new manuscript makes no change to the proof sources
covered by the completed independent replay.

%% file: code_table.tex
\begingroup\small
\begin{longtable}{r@{\qquad}r@{\qquad}r}
\caption{Exact code numerators; $w_{i,c}=N_{i,c}/2^{55}$.}\label{tab:code}\\
\toprule $i$ & $N_{i,0}$ & $N_{i,1}$ \\ \midrule
\endfirsthead
\multicolumn{3}{c}{Table \thetable\ continued}\\
\toprule $i$ & $N_{i,0}$ & $N_{i,1}$ \\ \midrule
\endhead
\midrule\multicolumn{3}{r}{Continued on next page}\\\endfoot
\bottomrule\endlastfoot
0 & 8140985237798599 & 100381493485016752 \\
1 & 66873931902426304 & 16802802218158854 \\
2 & 19879035082725536 & 59583145120762928 \\
3 & 82703681629426432 & 11839583596344262 \\
4 & 16253725455866406 & 68340830244727080 \\
5 & 39691610129809872 & 32607190098165324 \\
6 & 13637566269366960 & 76217273187066688 \\
7 & 85546870989978176 & 11136585844599798 \\
8 & 14432245113150018 & 73652856759482000 \\
9 & 30314556834137416 & 42449717367346168 \\
10 & 30330665508747488 & 42429379754296720 \\
11 & 64538358005372064 & 17722545547055206 \\
12 & 15011650986123108 & 71882938506889416 \\
13 & 53177359053108936 & 23156126773920168 \\
14 & 12102715404927432 & 81687641146737856 \\
15 & 91429160876521520 & 9824405767541486 \\
16 & 12255517331836568 & 81108917930130112 \\
17 & 26180352673880852 & 48189955664327472 \\
18 & 22362029444810328 & 54622984962053256 \\
19 & 46365696235847392 & 27410516726943052 \\
20 & 21709748934760232 & 55858540885961120 \\
21 & 35746989920010232 & 36312140296164672 \\
22 & 40492742859816456 & 31918081389777512 \\
23 & 73858156317735360 & 14366713149911624 \\
24 & 14701452053380894 & 72820548832190976 \\
25 & 51128121200576008 & 24342698114409384 \\
26 & 40211621795473736 & 32157647421268008 \\
27 & 56657164459430400 & 21300509752417452 \\
28 & 16141801286849876 & 68646935688668416 \\
29 & 60542399330254304 & 19438124008426808 \\
30 & 21275323130904828 & 56706926374105728 \\
31 & 90541051759959888 & 10011100686904910 \\
32 & 10179901224578372 & 89753600394574384 \\
33 & 56473662083282152 & 21393702844091528 \\
34 & 19177520510604684 & 61121727051520440 \\
35 & 68067766045868232 & 16354331693356560 \\
36 & 20143472611901056 & 59020023253528424 \\
37 & 29911154338323744 & 42963717108079136 \\
38 & 23657653026803248 & 52295606860683904 \\
39 & 70552891808698168 & 15464707175748434 \\
40 & 13884656710800832 & 75402036593447184 \\
41 & 29606341183199600 & 43358185484296368 \\
42 & 34161217403479312 & 37965987263737208 \\
43 & 50572492398644336 & 24677437257282220 \\
44 & 25009246527530016 & 50030991897680784 \\
45 & 52116956249653944 & 23760906523063720 \\
46 & 46524085925533016 & 27300810438840792 \\
47 & 80672814765149840 & 12372093880466334 \\
48 & 9873133277626622 & 91195587395410656 \\
49 & 80814422296251424 & 12334104181822422 \\
50 & 23452645434471648 & 52653206695368352 \\
51 & 69943199371999512 & 15677655158896560 \\
52 & 18700356812023332 & 62207262962683520 \\
53 & 40162817189987136 & 32199483061127640 \\
54 & 41494684372764712 & 31083385827777204 \\
55 & 70910997599424808 & 15341190978643718 \\
56 & 11541016200305046 & 83887581053610736 \\
57 & 74495464092136672 & 14165441842135428 \\
58 & 35465015548196632 & 36598760616713992 \\
59 & 65828165595572704 & 17207511417398896 \\
60 & 12589933945592770 & 79870137884132624 \\
61 & 58802627644078896 & 20246693290618700 \\
62 & 17819301423540656 & 64300980547725592 \\
63 & 100379572983432992 & 8141310872402782 \\
\end{longtable}
\endgroup

%% file: log_table.tex
\begingroup\footnotesize
\begin{longtable}{r@{\quad}r@{\quad}r}
\caption{Exact packed affine logarithm coefficients for \eqref{eq:logpair}.}\label{tab:logs}\\
\toprule $j$ & $C_j$ & $V_j$ \\ \midrule
\endfirsthead
\multicolumn{3}{c}{Table \thetable\ continued}\\
\toprule $j$ & $C_j$ & $V_j$ \\ \midrule
\endhead
\midrule\multicolumn{3}{r}{Continued on next page}\\\endfoot
\bottomrule\endlastfoot
0 & 3029110255766351240742567951 & 304723702261735822664128779 \\
1 & 626973476636423485346604672103 & 295627465185037053085388339 \\
2 & 1232299879142570971276078181558 & 287058565181173418458707781 \\
3 & 1820087069255880127219887041878 & 278972398022974689088292508 \\
4 & 2391322622363278445674133931514 & 271329321657426463147717584 \\
5 & 2946913039194004356790734393508 & 264093881442419070878775134 \\
6 & 3487692358087906409303199369574 & 257234293637913510724029067 \\
7 & 4014429617625758289725741863820 & 250722039577571827716679498 \\
8 & 4527835532524385841913716677158 & 244531367606667123322669097 \\
9 & 5028568172134788765222313786618 & 238639034827666523506960442 \\
10 & 5517238032226849728358099019438 & 233024011952325999380710151 \\
11 & 5994412608490592564188146453934 & 227667132813552966719788163 \\
12 & 6460620208205356964418398600853 & 222551020578429991126571262 \\
13 & 6916353572940526883191379805296 & 217659792600309608677117211 \\
14 & 7362072970654112404502366943658 & 212978931291605809954303432 \\
15 & 7798208953905263875472675863143 & 208495173443329597790517298 \\
16 & 8225164955827671911347441315920 & 204196288864160102753036743 \\
17 & 8643319380034990482345557886505 & 200071117273932730563133477 \\
18 & 9053027779587140388379336067363 & 196109310049222130162350366 \\
19 & 9454624546638279434554592189628 & 192301367116830341131604653 \\
20 & 9848424624276205727343328918902 & 188638489379834204014775029 \\
21 & 10234724977023008254414351955500 & 185112523377353139190046782 \\
22 & 10613805873676990002616618146604 & 181715961284549146869911790 \\
23 & 10985932185668302841952316483070 & 178441811785418291134307211 \\
24 & 11351354505817082107913754575952 & 175283544732558478801960628 \\
25 & 11710310079451300466046686077986 & 172235128040657606849493280 \\
26 & 12063023765334560199628273506111 & 169290953899517267573213600 \\
27 & 12409708959368355956599175925648 & 166445728093588306331807508 \\
28 & 12750568261755467663416902264297 & 163694543788947116384544885 \\
29 & 13085794196810361025276224982142 & 161032844639807491472176328 \\
30 & 13415569934282221040373776930393 & 158456332554800257269175077 \\
31 & 13740069885295466045168825898173 & 155960949250229197674185394 \\
\end{longtable}
\endgroup